\documentclass[final,onecolumn]{IEEEtran}

\def\final{0}

\usepackage{graphicx}
\usepackage{dcolumn}
\usepackage{bm}

\usepackage{fancyhdr}
\usepackage{latexsym}
\usepackage{graphicx,color}
\usepackage{epstopdf}  
\usepackage[pdfstartview=FitH]{hyperref}
\usepackage{xargs}
\usepackage{shadow,epsf,amsthm,amssymb,amsmath}

\usepackage{enumerate}

\DeclareSymbolFont{rsfscript}{OMS}{rsfs}{m}{n}
\DeclareSymbolFontAlphabet{\mathrsfs}{rsfscript}

\usepackage{setspace}                           

\usepackage{lipsum}

\usepackage{tikz}	
\usetikzlibrary{backgrounds,fit,decorations.pathreplacing}  

\usepackage[colorinlistoftodos,prependcaption]{todonotes}

\newtheorem{definitionenv}{Definition}
\newtheorem{lemmaenv}[definitionenv]{Lemma}
\newtheorem{theoremenv}[definitionenv]{Theorem}
\newtheorem{corollaryenv}[definitionenv]{Corollary}
\newtheorem{propositionenv}[definitionenv]{Proposition}
\newtheorem{conjectureenv}[definitionenv]{Conjecture}
\newtheorem{remarkenv}[definitionenv]{Remark}
\newenvironment{remark}{\begin{remarkenv}\rm}{\end{remarkenv}}
\newcommand{\br}{\begin{remark}}
\newcommand{\er}{\end{remark}}

\newtheorem{exampleenv}{Example}
\newtheorem{app-lemmaenv}[section]{Lemma}

\newenvironment{definition}{\begin{definitionenv}\rm}{\end{definitionenv}}
\newenvironment{lemma}{\begin{lemmaenv}\rm}{\end{lemmaenv}}
\newenvironment{theorem}{\begin{theoremenv}\rm}{\end{theoremenv}}
\newenvironment{corollary}{\begin{corollaryenv}\rm}{\end{corollaryenv}}
\newenvironment{example}{\begin{exampleenv}\rm}{\end{exampleenv}}
\newenvironment{proposition}{\begin{propositionenv}\rm}{\end{propositionenv}}
\newenvironment{conjecture}{\begin{conjectureenv}\rm}{\end{conjectureenv}}
\newenvironment{app-lemma}{\begin{app-lemmaenv}\rm}{\end{app-lemmaenv}}

\newcommand{\bd}{\begin{definition}}
\newcommand{\ed}{\end{definition}}
\newcommand{\bl}{\begin{lemma}}
\newcommand{\el}{\end{lemma}}
\newcommand{\elp}{\hspace*{\fill} $\Box$
                 \end{lemma}}
\newcommand{\bt}{\begin{theorem}}
\newcommand{\et}{\end{theorem}}
\newcommand{\etp}{\hspace*{\fill} $\Box$
                 \end{theorem}}
\newcommand{\bc}{\begin{corollary}}
\newcommand{\ec}{\end{corollary}}
\newcommand{\ecp}{\hspace*{\fill} $\Box$
                 \end{corollary}}
\newcommand{\bcj}{\begin{conjecture}}
\newcommand{\ecj}{\end{conjecture}}

\newcommand{\be}{\begin{example}}
\newcommand{\ee}{\end{example}}
\newcommand{\eep}{\hspace*{\fill} $\Box$
                 \end{example}}
\newcommand{\bp}{\begin{proposition}}
\newcommand{\ep}{\end{proposition}}
\newcommand{\epp}{
                 \end{proposition}}

\newcommand{\bra}[1]{\langle#1|}
\newcommand{\ket}[1]{|#1\rangle}
\newcommand{\braket}[2]{\langle#1|#2\rangle}

\newcommand{\tr}[1]{\text{tr}\left(#1\right)}

\newcommand{\eeq}{ \setcounter{equation} {\value{enumi}}}

\newcommand{\ot}{\otimes}
\newcommand{\I}{\mathsf{id}}

\newcommand{\Hmin}{H_{\rm min}}

\newcommand{\cA}{\mathcal{A}}
\newcommand{\cB}{\mathcal{B}}

\newcommand{\cE}{\mathcal{E}}

\newcommand{\cR}{\mathcal{R}}

\newcommand{\cX}{\mathcal{X}}

\newcommand{\cY}{\mathcal{Y}}
\newcommand{\cZ}{\mathcal{Z}}

\newcommand{\rA}{\mathrsfs{A}}
\newcommand{\rB}{\mathrsfs{B}}

\def\tr{\textnormal{tr}}
\def\beq{\begin{equation}}
\def\eeq{\end{equation}}

\def\bean{\begin{IEEEeqnarray*}{rCl}}
\def\eean{\end{IEEEeqnarray*}}

\newcommand{\id}{\text{id}}

\newcommand{\etal}{\emph{et~al.}}

\newcommand{\zo}{\{0,1\}}

\ifnum\final=0
\newcommand{\mynote}[2]{{\color{#1} \marginpar{\tiny #2}}}
\newcommand{\mybignote}[2]{{\color{#1} $\langle \langle$ #2$\rangle \rangle$}}
\newcommandx{\rednote}[2][1=]{\todo[linecolor=red,backgroundcolor=red!25,bordercolor=red,#1]{#2}}
\newcommandx{\bluenote}[2][1=]{\todo[linecolor=blue,backgroundcolor=blue!25,bordercolor=blue,#1]{#2}}
\newcommandx{\yellownote}[2][1=]{\todo[linecolor=yellow,backgroundcolor=yellow!25,bordercolor=yellow,#1]{#2}}
\newcommandx{\greennote}[2][1=]{\todo[inline,linecolor=olive,backgroundcolor=green!25,bordercolor=olive,#1]{#2}}

\else
\newcommand{\mynote}[2]{}
\newcommand{\mybignote}[2]{}
\newcommand{\rednote}[2][1=]{}
\newcommand{\bluenote}[2][1=]{}
\newcommand{\greennote}[2][1=]{}
\newcommand{\yellownote}[2][1=]{}

\fi

\begin{document}
\title{Interactive Leakage Chain Rule for Quantum Min-entropy}
\author{Ching-Yi Lai and Kai-Min Chung
\thanks{

CYL is with the  Institute of Communications Engineering,
        National Chiao Tung University,
        Hsinchu 30010, Taiwan.
(email: cylai@nctu.edu.tw)

KMC is with the Institute of Information Science, Academia Sinica,
Nankang, Taipei 11529, Taiwan.
(email: kmchung@iis.sinica.edu.tw)
}
}



\maketitle

\begin{abstract}
The leakage chain rule for quantum min-entropy quantifies the change of min-entropy when one party gets additional leakage
about the information  source. Herein we provide an interactive version that quantifies the change of min-entropy between two parties, who share an initial
classical-quantum state and are allowed to run a two-party protocol. As an application, we prove new versions of lower bounds on the complexity of quantum communication of classical information.
\end{abstract}

\section{Introduction}

Let $(X,Y,Z)$ be a classical distribution over $\zo^{n}\times\zo^m\times \zo^{\ell}$. (Classical) leakage chain rule states that
$$H(X|Y,Z) \geq H(X|Y) -  \ell,$$
which says that an $\ell$-bit ``leakage'' $Z$ can decrease the  entropy of $X$ (conditioned on $Y$) by at most $\ell$. Note that the statement is different from the standard chain rule for Shannon entropy that $H(X,Y) = H(X) + H(Y|X)$. Leakage chain rule generally holds  for various entropy notions and is especially useful for cryptographic applications. In particular, a computational leakage chain rule for computational min-entropy, first proved by~\cite{DziembowskiP08,ReingoldTTV08}, has found several applications in classical cryptography~\cite{GentryW11,ChungKLR11,FOR12,JP14,ChungLP15}.

The notion of (smooth) min- and max- entropies  in the quantum setting are proposed by Renner and Wolf~\cite{RW04}. The  leakage chain rule for quantum min-entropy has also been discussed and is more complicated than its classical analogue due to the effect of quantum entanglement. Consider a state $\rho_{XYZ}$ on the state space $\cX \ot \cY \ot \cZ$, where $Z$ is an $\ell$-qubit system. The  leakage chain for quantum min-entropy states that
\begin{equation} \label{eqn:LCR}
 \Hmin(X|Y,Z)_\rho \geq
\begin{cases}
\Hmin(X|Y)_\rho - \ell,  \mbox{\quad \ \ if $\rho$ is a separable state on $(\cX \ot \cY)$ and $\cZ$} \\
\Hmin(X|Y)_\rho - 2\ell, \mbox{\quad otherwise.}
\end{cases}
\end{equation}
In other words, the leakage $Z$ can decrease the quantum min-entropy of $X$ conditioned on $Y$ by at most $\ell$ if there is no entanglement, and $2\ell$ in general. Note that the factor of $2$ is tight by
 the application of superdense coding~\cite{BW92}. The separable case is proved by Desrosiers and Dupuis~\cite{DD10}, and the general case is proved by Winkler \etal~\cite{WTHR11}, both of which are motivated by cryptographic applications. Furthermore, a computational version of quantum leakage chain rule is explored in~\cite{CCL+17} with applications in quantum leakage-resilient cryptography.

Herein we formulate an interactive version of leakage chain rule with initial classical-quantum (cq) states. Let $\rho_{XY}$ be a cq-state shared between Alice and Bob. Consider that $X$ is a classical input from Alice. Then Alice and Bob engage in an interaction where Alice may leak information about $X$ to Bob. We are interested in how much leakage is generated from the interaction regarding to the communication complexity of the interaction.
We restrict the discussion to the situation where $X$ is a classical input that remains constant during the interaction.
 This is formalized by allowing Alice to perform only quantum operations controlled by $X$ on her system.

\begin{theorem}[Interactive leakage chain rule for quantum min-entropy] \label{thm:LCR-informal}
  Suppose Alice and Bob share a cq  state $\rho = \rho_{XY} \in D(\cX\ot\cY)$,
  where Alice holds the classical system $X$  and Bob holds the quantum system $Z$. 
  If an interactive protocol $\Pi$ is executed by Alice and Bob
 and $m_{B}$ and $m_{A}$ are the total numbers of qubits that Bob and Alice send to each other, respectively,
  then
    \begin{align}
    &\Hmin(X|Y)_{\sigma}\geq \Hmin(X|Y)_\rho - \min\{ m_{B}+m_{A},  2m_{A}\}, %
    \end{align}
    where    $\sigma_{XY}=\Pi(\rho_{XY})$ is the joint state at the end of the protocol.
\end{theorem}

It is interesting to discuss the implication of Theorem~\ref{thm:LCR-informal} to Holevo's problem of conveying classical messages by transmitting quantum states. In the interactive setting, Cleve \etal~\cite{CDNT99} and Nayak and Salzman~\cite{NS06} showed that for Alice to reliably communicate $n$ bits of classical information to Bob, roughly $n$ qubits of total communication and $n/2$ qubits of one-way communication from Alice to Bob are necessary. The same conclusion follows immediately from Theorem~\ref{thm:LCR-informal}.

In fact, in the case without initial shared cq-states,  the general form of the result in~\cite{NS06} (Theorem 1.4) agrees  to the above interactive leakage chain rule.
Thus our interactive leakage chain rule can be viewed as a generalization of~\cite{NS06} to allow initial correlation between $X$ and $Y$. We remark that our proof is not a generalization of the proof in~\cite{NS06}, although we both used Yao's lemma~\cite{Yao93}.
Conceptually, the use of interactive leakage chain rule makes the proof simple.

This manuscript is organized as follows.
In Sec.~\ref{sec:prelim} we give some basics about quantum information.
Then we discuss the leakage chain rule for quantum min-entropy and its application to the problem of communicating classical information in Sec.~\ref{sec:leakage-chain-rule}.


\section{Preliminaries} \label{sec:prelim}

We give notation and briefly introduce basics of quantum mechanics here.
The Hilbert space of a quantum system $A$  is denoted by the corresponding calligraphic letter  $\cA$
and its dimension is denoted by   $d_A$.
Let $L(\cA)$ be the space of linear operators on $\cA$.
A  quantum state of system $A$ is described by a \emph{density operator} $\rho_A\in L(\cA)$ that is positive semidefinite  and with unit trace $(\tr(\rho_A)=1)$. 
Let $D(\cA)= \{ \rho_A\in L(\cA): \rho_A\geq 0, \tr(\rho_A)=1\}$ be the set of density operators on  $\cA$.
When $\rho_A\in D(\cA)$ is of rank one, it is called a \emph{pure} quantum state and we can write $\rho=\ket{\psi}_A\bra{\psi}$ for some unit vector $\ket{\psi}_A\in \cA$,
where $\bra{\psi}=\ket{\psi}^{\dag}$ is the conjugate transpose of $\ket{\psi}$. If  $\rho_A$ is not pure, it is called a \emph{mixed} state and can be expressed as a convex combination of pure quantum states.

The evolution of a quantum state $\rho\in D(\cA)$ is described by a completely positive and trace-preserving (CPTP) map $\Psi: D(\cA)\rightarrow D(\cA')$  
such that $\Psi(\rho) =\sum_{k} E_k \rho E_k^\dag$,
where $\sum_k  E_k^\dag E_k =\I_{A}$.
In particular, if the evolution is a unitary $U$, we have the evolved state $\Psi(\rho)=U\rho U^{\dag}$.

The Hilbert space of a joint quantum system $AB$ is the tensor product of the corresponding Hilbert spaces $\cA\otimes \cB$.
Let $\I_A$ denote the identity on system $A$.
For $\rho_{AB}\in D(\cA\ot \cB)$, we will use  $\rho_A=\tr_B(\rho_{AB})$
to denote its reduced density operator in system $A$,
where
\[
\tr_B(\rho_{AB})= \sum_{i} \I_A\ot\bra{i}_B \rho_{AB} \I_A\ot \ket{i}_B
\]
for an orthonormal basis $\{\ket{i}_B\}$ for $\cB$.
A \emph{separable} state $\rho_{AB}$ has a density operator of the form
\[
\rho_{AB}=\sum_{x} p_x \rho_A^x\otimes \rho_B^x,
\]
where $\rho_A^x \in D(\cA)$ and $\rho_B^x \in D(\cB)$. In particular,
a classical-quantum (cq) state $\rho_{AB}$ has a density operator of the form
\[
\rho_{AB}=\sum_{a} p_a \ket{a}_A\bra{a}\otimes \rho_B^a,
\]
where $\{\ket{a}_A\}$ is an orthonormal basis for $\cA$ and $\rho_B^a \in D(\cB)$.
We define the following specific quantum operations on cq-states that preserve the classical system.
\bd A  quantum operation $\Gamma$ on a classical-quantum system $AB$ is said to be controlled by the classical system $A$ if, for a cq state $\rho_{AB}= \sum_{a} p_a  \ket{a}_A\bra{a}\otimes \rho_B^a$, $$\Gamma(\rho_{AB})= \sum_{a} p_a  \ket{a}_A\bra{a}\otimes \Gamma^a(\rho_B^a),$$
where $\Gamma^a$ are CPTP maps.
In this case, $\Gamma$ is called a \emph{classically-controlled quantum operation}.
In particular, if $\Gamma^a$ are unitaries, $\Gamma$ is called a  \emph{classically-controlled unitary}.
\ed
\noindent Note that the reduced state for classical system $A$ of a cq-state $\rho_{AB}$  remains the same after a
classically-controlled quantum operation $\Gamma$. That is, $\tr_B \rho_{AB}= \tr_B{\Gamma(\rho_{AB})}$.

\bl [Schmidt decomposition]
For a pure state $\ket{\psi}_{AB}\in\cA\ot\cB$, there exist orthonormal states $\{\ket{i}_A\}\in\cA$ and  $\{\ket{i}_B\}\in\cB$
such that \[
\ket{\psi}_{AB}=\sum_{i=1}^s \lambda_i \ket{i}_A\ot \ket{i}_B,
\]
where $\lambda_i\geq 0$, $s\leq \min\{d_A, d_B\}$, and the smallest such $s$ is called the \emph{Schmidt rank} of $\ket{\psi}_{AB}$.
\el
\bl [Purification]
Suppose $\rho_A\in D(\cA)$ of finite dimension $d_A$. Then there exists $\cB$ of dimension $d_B\geq d_A$ and $\ket{\psi}_{AB}\in \cA\otimes \cB$ such that
\[
\tr_B \ket{\psi}_{AB}\bra{\psi} = \rho_A.
\]
\el



The trace distance between two quantum states $\rho$ and $\sigma$ is
$$||{\rho}-{\sigma}||_{\mathrm{tr}},$$
where $||X||_{\mathrm{tr}}=\frac{1}{2}\tr{\sqrt{X^{\dag}X}}$ is the trace norm of $X$.
The fidelity between $\rho$ and $\sigma$ is
$$
F(\rho,\sigma)=\tr \sqrt{\rho^{1/2}{\sigma}\rho^{1/2}}.
$$
\bt[Uhlmann's theorem~\cite{Uhl76}]
$$F(\rho_{A},\sigma_A)= \max_{\ket{\phi'}} |\braket{\psi}{\phi'}|,$$ where the maximization is over all purification of $\sigma_A$.
\et
\noindent Below is a variant of Uhlmann's theorem.
\bc \label{cor:Uhlmann}
Suppose $\rho_A$ is a reduced density operator of $\rho_{AB}$.
Suppose $\rho_A$ and $\sigma_A$ have fidelity $F(\rho_A,\sigma_A)\geq 1- \epsilon$. Then
there exists $\sigma_{AB}$ with $\tr_B(\sigma_{AB})=\sigma_A$ such that $F(\rho_{AB},\sigma_{AB})\geq 1- \epsilon$.
\ec
\begin{proof}
Let $\ket{\psi}_{ABR}$ be a purification of $\rho_{AB}$, which is immediately  a purification of $\rho_A$.
Suppose $\ket{\phi}$ is a purification of $\sigma_A$ such that $|\braket{\psi}{\phi}|\geq 1-\epsilon$. Let $\sigma_{AB}= \tr_R(\ket{\phi}\bra{\phi})$.
Then $F(\rho_{AB},\sigma_{AB})\geq |\braket{\psi}{\phi'}| \geq 1- \epsilon$.
\end{proof}
A relation between the fidelity and the trace distance of two quantum states $\sigma$ and $\rho$ was proved by Fuchs and van~de Graaf~\cite{FvdG99} that
\begin{align}
1-F(\rho,\sigma)\leq   || \rho-\sigma||_{\mathrm{tr}}\leq \sqrt{1- F^2(\rho,\sigma)}. \label{eq:FT}
\end{align}
The purified distance is defined as
\begin{align}
P(\rho,\sigma)=&\sqrt{1- F^2(\rho,\sigma)}.
\end{align}

For a {one-sided two-party protocol} (that is, only one party will have the output), where Alice has no (or little) information about Bob's input, Lo showed that it is possible for Bob to cheat by changing his input at a later time~\cite{Lo97}.
The basic idea can be formulated as the following lemma, which is proved by a standard argument using Uhlmann theorem and the Fuchs and van~de Graaf~inequality~\cite{FvdG99} (for a proof, see, e.g.,~\cite{BB15}).
\bl\label{lemma:lo_attack}
Suppose $\rho_A$, $\sigma_A\in\cA$ are two quantum states with  purifications $\ket{\phi}_{AB}$, $\ket{\psi}_{AB}\in \cA\ot\cB$, respectively, and $||\rho_A-\sigma_A||_{\tr}\leq \epsilon$. Then there exists a unitary $U_B\in L(\cB)$ such that
$$||\ket{\phi}_{AB}-\I_A\ot U_B\ket{\psi}_{AB}||_{\tr}\leq \sqrt{\epsilon(2-\epsilon)}.$$
\el

%
%
%
%

\subsection{Protocol Definition}

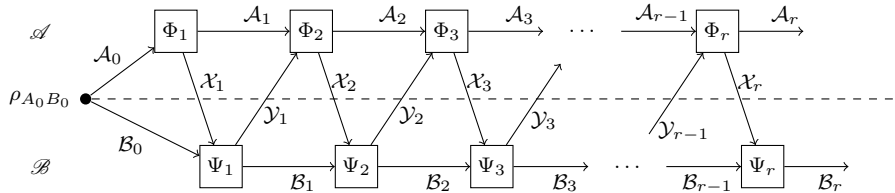
\begin{figure} [h]
\centerline{
\vspace{-0.3cm}
    \begin{tikzpicture}[scale=0.6][thick]
   \fontsize{8pt}{1} 
    \tikzstyle{variablenode} = [draw,fill=white, shape=rectangle,minimum size=2.em]
    \node [fill=black,inner sep=1.5pt,shape=circle,label=left:$\rho_{A_0B_0}$] (n0) at (-2,1.5) {} ;
        \node [fill=white,inner sep=1.5pt] (n1) at (16,1.5) {} ;
    \node [fill=white,inner sep=1.5pt] (bn4) at (-3,3) {$\mathrsfs{A}$} ;
    \node [fill=white,inner sep=1.5pt] (bn4) at (-3,0) {$\mathrsfs{B}$} ;
    \node[variablenode] [fill=white,inner sep=1.5pt ](an1)  at(0,3) {$\Phi_1$};
    \node[variablenode] [fill=white,inner sep=1.5pt] (an2) at (3,3) {$\Phi_2$} ;
    \node[variablenode] [fill=white,inner sep=1.5pt](an3)  at(6,3) {$\Phi_3$};
    \node [fill=white,inner sep=1.5pt,shape=circle] (an4) at (9,3) {\quad$\dots \quad$} ;
        \node[variablenode] [fill=white,inner sep=1.5pt] (an5) at (12,3) {$\Phi_r$} ;
            \node [fill=white,inner sep=1.5pt] (an6) at (14,3) {$$} ;
    \node[variablenode] [fill=white,inner sep=1.5pt] (bn1) at (1,0) {$\Psi_1$} ;
    \node[variablenode] [fill=white,inner sep=1.5pt] (bn2) at (4,0) {$\Psi_2$} ;
    \node[variablenode] [fill=white,inner sep=1.5pt] (bn3) at (7,0) {$\Psi_3$} ;
    \node [fill=white,inner sep=1.5pt,shape=circle] (bn4) at (10,0) {\quad$\dots \quad$} ;
    \node[variablenode] [fill=white,inner sep=1.5pt] (bn5) at (13,0) {$\Psi_r$} ;
    \node [fill=white,inner sep=1.5pt] (bn6) at (15,0) {} ;
      \draw[dashed]  (n0) -- (n1);
    \draw[->] (n0)--(an1) node [above,midway] {$\cA_0$\mbox{  \quad}};
    \draw[->] (n0)--(bn1) node [below,midway] {$\cB_0$\mbox{  \quad}};
    \draw[->] (an1) -- (bn1) node [above,midway] {\mbox{  \quad}$\cX_1$};
    \draw[->]    (an2) -- (bn2)  node [above,midway] {\mbox{\quad}$ \cX_2$};
    \draw[->]    (an3) -- (bn3)  node [above,midway] {\mbox{\quad}$ \cX_3$};
    \draw[->]    (an5) -- (bn5)  node [above,midway] {\mbox{\quad}$ \cX_r$};
    \draw[->] (an1) -- (an2) node [above,midway] {\mbox{  \quad}$\cA_1$};
    \draw[->] (an2) -- (an3) node [above,midway] {\mbox{  \quad}$\cA_2$};
    \draw[->] (an3) -- (an4) node [above,midway] {\mbox{  \quad}$\cA_3$};
    \draw[->] (an4) -- (an5) node [above,midway] {\mbox{  \quad}$\cA_{r-1}\quad $};
    \draw[->] (an5) -- (an6) node [above,midway] {\mbox{  \quad}$\cA_{r}$};
    \draw[->] (bn1) -- (an2) node [below,midway] {\quad$ \cY_1$} ;
    \draw[->] (bn2) --(an3) node [below,midway] {\mbox{\quad}$ \cY_2$} ;
    \draw[->] (bn3) --(an4) node [below,midway] {\mbox{\quad}$ \cY_3$} ;
    \draw[->] (bn4) --(an5) node [below,near start] {\mbox{\quad}$\quad \cY_{r-1}$} ;
    \draw[->] (bn1) -- (bn2) node [below,midway] {\mbox{  \quad}$\cB_1$};
    \draw[->] (bn2) -- (bn3) node [below,midway] {\mbox{  \quad}$\cB_2$};
    \draw[->] (bn3) -- (bn4) node [below,midway] {\mbox{  \quad}$\cB_3$};
    \draw[->] (bn4) -- (bn5) node [below,midway] {\mbox{  \quad}$\cB_{r-1}\quad $ };
    \draw[->] (bn5) -- (bn6) node [below,midway] {\mbox{  \quad}$\cB_{r}$};
   \end{tikzpicture}
 }
 \caption{An interactive two-party  quantum protocol.} \label{fig:two-party}
\end{figure}
 We basically follow  the definition  of two-party quantum protocol~\cite{GW07,DNS10}.
Consider a quantum protocol between two parties $\rA$ and $\rB$, where the party $\rA$ sends the first and the last messages without loss of generality.
Such a two-party quantum protocol is defined as follows.
 \bd (Two-party quantum protocol) \label{def:two-party}
 An $(r,m_A,m_B)$ protocol $\Pi=( \mathrsfs{A},\mathrsfs{B})$ is a two-party protocol with  $r$ rounds of interaction defined as follows:
 \begin{enumerate}
 \item input spaces $\cA_0$ and $\cB_0$ for parties $\rA$ and $\rB$, respectively;
 \item memory spaces $\cA_1, \dots, \cA_r$ for $\rA$ and  $\cB_1, \dots, \cB_r$ for $\rB$;
 \item communication spaces  $\cX_1, \dots, \cX_r$,  $\cY_1, \dots, \cY_{r-1}$;
 \item a series of quantum operations $\Phi_1,\dots,\Phi_r$ for $\rA$
 and a series of quantum operations $\Psi_1,\dots,\Psi_r$ for $\rB$,
 where
 \begin{align*}
 \Phi_1:& L(\cA_0)\rightarrow L(\cA_1\otimes \cX_1);\\
 \Phi_i:& L(\cA_{i-1}\otimes \cY_{i-1})\rightarrow L(\cA_{i}\otimes \cX_{i}),\  i=2,\dots,r;\\
 \Psi_j:& L(\cB_{j-1}\otimes \cX_{j})\rightarrow L(\cB_{j}\otimes \cY_{j}),\  j=1,\dots,r-1;\\
 \Psi_r:& L(\cB_{r-1}\otimes \cX_{r})\rightarrow L(\cB_r).
 \end{align*}
\end{enumerate}
The \emph{one-way communication complexities} (in terms of qubits) sent from Alice to Bob and from Bob to Alice are
$m_A=\sum_{i=1}^r \log d_{X_i}$ and $m_B=\sum_{j=1}^{r-1}\log d_{Y_j}$, respectively.
The  \emph{(total) communication complexity} of this protocol is $m_A+m_B$.
 \ed
For input state $\rho\in D(\cA_0\otimes \cB_0\otimes \cR)$, where $R$ is a reference system of dimension $d_R= d_{A_0}d_{B_0}$,
let
\begin{align*}
[\rA_1^i\circledast \rB_1^{i-1}](\rho)=&  \left(\Phi_{i}\ot\I_{B_{i-1},R}\right)\left(\Psi_{i-1}\ot\I_{A_{i-1},R}\right) \cdots \left(\Psi_1\ot \I_{A_1,R}\right)\left(\Phi_1\ot \I_{B_0,R}\right) (\rho),\\
[\rA_1^i\circledast \rB_1^i](\rho)=&  \left(\Psi_{i}\ot\I_{A_{i},R}\right) \left(\Phi_{i}\ot\I_{B_{i-1},R}\right) \cdots \left(\Psi_1\ot \I_{A_1,R}\right)\left(\Phi_1\ot \I_{B_0,R}\right) (\rho),
\end{align*}
and let $\Pi(\rho)=[\rA_1^r\circledast \rB_1^r](\rho)$ denote the final state of protocol $\Pi=(\rA,\rB)$ on input $\rho$.

Figure~\ref{fig:two-party} illustrates an interactive two-party quantum protocol.
Note that the input state $\rho_{A_0B_0}\in D(\cA_0\otimes \cB_0)$  may consist of a classical string, tensor products of pure quantum states, or an entangled quantum state, depending on the context of the underlying protocol.
For example, a part of it can be EPR pairs shared between Alice and Bob.
Also the reference system $R$ is not shown in Fig.~\ref{fig:two-party}.

\begin{remark}
In the following discussion we will consider a specific two-party protocol, where the input of $\rA$ is a classical system $A_0$ that is preserved throughout the protocol and its quantum operations
 \begin{align*}
 \Phi_1:& L(\cA_0)\rightarrow L(\cA_0\otimes \cA_1\otimes \cX_1),\\
 \Phi_i:& L(\cA_0\otimes \cA_{i-1}\otimes \cY_{i-1})\rightarrow L(\cA_0\otimes \cA_{i}\otimes \cX_{i}),\  i=2,\dots,r
 \end{align*}
 are classically-controlled quantum operations controlled by $A_0$.
\end{remark}

\section{Leakage Chain Rule for quantum min-entropy}\label{sec:leakage-chain-rule}
We first review the notion of quantum (smooth) min-entropy~\cite{RW04}.
\bd \label{def:minentropy}
Consider a bipartite quantum state $\rho_{AB}\in D\left(\cA\otimes \cB\right)$.
The   min-entropy of $A$ conditioned on $B$ is defined as
\begin{align}
    H_{\min}(A|B)_{\rho} =  -\inf_{\sigma_{B}}
    \left\{
      \inf \left\{\lambda\in \mathbb{R}: \rho_{AB}\leq 2^\lambda \I_{A}\otimes \sigma_B\right\} \label{def:Hmin}
    \right\}.
 \end{align}
 \ed
When $\rho_{AB}$ is a cq-state, the quantum min-entropy has an operational meaning in terms of guessing probability~\cite{KRS09}. Specifically, if $H_{\min}(A|B)_{\rho} = k$, then the optimal probability of predicting the value of $A$ given $\rho_B$  is exactly $2^{-k}$.

 The smooth min-entropy of $A$ conditioned on $B$ is defined as
\begin{align*}
    H_{\min}^{\epsilon}(A|B)_{\rho} =  \sup_{ \rho':  P(\rho',\rho)<\epsilon} H_{\min}(A|B)_\rho.
 \end{align*}
{For simplicity we focus on the discussion of min-entropy and our results can be generalized to smooth min-entropy without much effort.}

In cryptography, we would like to see how much (conditional)  min-entropy is left in an information source  when the adversary gains additional information leakage.
This is characterized by the leakage chain rule  for min-entropy. In the quantum case, the situation is different due to the phenomenon of quantum entanglement.
When two parties share a separable quantum state $\rho$, this is like the classical case and we have the following leakage chain rule for conditional quantum min-entropy~\cite{DD10}:
\begin{lemma}\cite[Lemma 7]{DD10}
  Let $\rho = \rho_{AXB}=\sum_{k} p_{k}\rho_{AX}^k  \ot \rho_B^k$ be a separable state in $D(\cA\ot\cX\ot\cB)$.
  Then
  \[\Hmin(A|XB)_{\rho} \geq \Hmin(A|B)_\rho - \log d_X. \]
\end{lemma}

Winkler~\etal~\cite{WTHR11} proved  the leakage chain rule for quantum (smooth) min-entropy for general quantum states with entanglement. 
\bl{\cite[Lemma~13]{WTHR11}} \label{lemma:LCR2}
  Let $\rho = \rho_{AXB}$ be a quantum state in $D(\cA\ot\cX\ot\cB)$.
  Then
  \[\Hmin(A|XB)_{\rho} \geq \Hmin(A|B)_\rho - 2\log d.\]
where $d= \min\{d_A d_B, d_X\}$.
\el

Lemma~\ref{lemma:LCR2} only characterizes the entropy loss regarding the one-way communication complexity.
We would like to find one that characterizes the two-way communication complexity.
First we prove a variant of Yao's lemma~\cite{Yao93} (see also \cite[Lemma 4]{Kre95}).
For our purpose, the formulation is not symmetric in $\rA$ and $\rB$.

\bl \label{lemma:Yao}
Suppose  $\Pi=(\rA,\rB)$ is an $(r,m_A,m_B)$ quantum protocol with initial state $(\ket{x}\ket{0})_{A_0}\otimes \ket{\zeta}_{B_0}$,
where  $x$ is a binary string,
and that the quantum operations $\Phi_i$ for $\rA$ are classical-controlled unitaries controlled by $\ket{x}_{A_0}$, respectively.
Then the final state of the protocol can be written as
\[
 \sum_{i\in\{0,1\}^{m_A+m_B}} \lambda_{i}  \ket{x}_{A_0}\otimes\ket{\xi_{i}}_{A_r}\ot \ket{\zeta_{i}}_{B_r},
\]
where $\lambda_{i}\geq 0$;
 $\ket{\xi_{i}}_{A_r}$ can be determined by  $\Pi$   and $x$; and
 $\ket{\zeta_{i}}_{B_r}$ can be determined by  $\Pi$ and $\ket{\zeta}_{B_0}$.

\el
\begin{proof}
We prove it by induction.
For simplicity, we will ignore the fixed register $\ket{x}_{A_0}$ in the following and remember that $\Phi_i$ are classically-controlled unitaries controlled by $\ket{x}_{A_0}$.
Suppose $\Pi=(\rA,\rB)$ is defined as in Def.~\ref{def:two-party}.
Let    $m_{B}^{(i)} =\sum_{j=1}^i \log d_{Y_i}$ and $m_{A}^{(i)}=\sum_{j=1}^i \log d_{X_i}$.
The statement is true for the initial state $\rho=\ket{0}_{A_0}\otimes \ket{\zeta}_{B_0}$, which is of rank one.
Suppose the statement holds  after $k$ rounds. That is,
\[
[\rA_1^k\circledast \rB_1^k](\rho)= \sum_{i\in\{0,1\}^{m_{A}^{(k)}+m_{B}^{(k)}}} \lambda_{i}^{(k)}\ket{\xi_{i}^{(k)}}_{A_k Y_k}\ot \ket{\zeta_{i}^{(k)}}_{B_k},
\]
where we use the superscript $(k)$ to indicate the states $\ket{\xi_i}, \ket{\zeta_i}$ or coefficients $\lambda_i$ after $k$ rounds.

Thus
\begin{align*}\displaystyle
&(\Phi_{k+1}\otimes \id_{B_k}) \sum_{i\in\{0,1\}^{m_{A}^{(k)}+m_{B}^{(k)}}} \lambda_{i}^{(k)}\ket{\xi_{i}^{(k)}}_{A_k Y_k}\ot \ket{\zeta_{i}^{(k)}}_{B_k}\\
\stackrel{(a)}{=}&\sum_{i\in\{0,1\}^{m_{A}^{(k)}+m_{B}^{(k)}}} \lambda_{i}^{(k)} \sum_{a\in\{0,1\}^{\log d_{X_{k+1}}}} \alpha_{i,a}  \ket{\xi_{i}^{(k+1)}, a}_{A_{k+1}}\ot \ket{a}_{X_{k+1}}  \ot \ket{\zeta_{i}^{(k)}}_{B_k}\\
\stackrel{\Psi_{k+1}\otimes \id_{A_{k+1}}}{\longrightarrow}& \sum_{i\in\{0,1\}^{m_{A}^{(k)}+m_{B}^{(k)}}} \sum_{a\in\{0,1\}^{\log d_{X_{k+1}}}} \lambda_{i}^{(k)}\alpha_{i,a}  \ket{\xi_{i}^{(k+1)}, a}_{A_{k+1}}\ot  \Psi_{X_{k+1},B_{k}}\left(\ket{a}_{X_{k+1}}  \ket{\zeta_{i}^{(k)}}_{B_k}\right)\\
\stackrel{(b)}{=}&\sum_{i\in\{0,1\}^{m_{A}^{(k)}+m_{B}^{(k)}}} \sum_{a\in\{0,1\}^{\log d_{X_{k+1}}}} \lambda_{i}^{(k)}\alpha_{i,a}  \ket{\xi_{i}^{(k+1)}, a}_{A_{k+1}}\ot \sum_{b\in\{0,1\}^{\log d_{Y_{k+1}}}} \beta_{i,a,b} \ket{b}_{Y_{k+1}}  \ot \ket{\zeta_{i}^{(k)},a,b}_{B_{k+1}}\\
\stackrel{(c)}{=}&\sum_{i\in\{0,1\}^{m_{A}^{(k+1)}+m_{B}^{(k+1)}}} \lambda_{i}^{(k+1)}\ket{\xi_{i}^{(k+1)}}_{A_{k+1}Y_{k+1}}\ot \ket{\zeta_{i}^{(k+1)}}_{B_{k+1}},
\end{align*}
where $(a)$ and $(b)$ are by Schmidt decomposition on $ \Phi_{k+1} \ket{\xi_{i}^{(k)}}_{A_kY_k}$ and $\Psi_{k+1}\left(\ket{a}_{X_{k+1}}  \ket{\zeta_{i}^{(k)}}_{B_k}\right)$, respectively,
with $\alpha_{i,a},\beta_{i,a,b}>0$;
in (c) the indexes   $i$,  $a$, and $b$ are merged and $ \lambda_{i:a:b}^{(k+1)}= \lambda_{i}^{(k)}\alpha_{i,a}\beta_{i,a,b}$.
(We use $a:b$ to denote the concatenation of two strings $a$ and $b$.)

Since $\sum_{b\in\{0,1\}^{d_{Y_{k+1}}}} \beta_{i,a,b} \ket{b}_{Y_{k+1}}  \ot \ket{\zeta_{i:a:b}^{(k+1)}}_{B_{k+1}}=\Psi_{{k+1}}\left(\ket{a}_{X_{k+1}}  \ket{\zeta_{i}^{(k)}}_{B_k}\right)$ and $\ket{\zeta_{i}^{(k)}}_{B_k}$ can be determined by $\Pi$ and $\ket{\zeta}_{B_0}$ by assumption,
 $\ket{\zeta_{i:a:b}^{(k+1)}}_{B_k}$ can also be determined by $\Pi$ and $\ket{\zeta}_{B_0}$.
 Similarly, $\ket{\xi_i^{(k+1)}}_{A_k}$ can be generated by $\Pi$ and $x$.

\end{proof}

Next we consider a special type of interactive two-party protocol on an input cq-state $\rho = \rho_{AB}$, where the system $A$ is classical and will be preserved throughout the protocol. The interactive leakage chain rule bounds how much the min-entropy $H_{\min}(A|B)_{\rho}$ can be decreased by an ``interactive leakage'' generated by applying a two-party protocol $\Pi =  \{\rA,\rB\}$ to $\rho$, where $A$ is treated as a classical input to $\rA$ and $B$ is given to $\rB$ as part of its initial state. 


\begin{theorem}[Interactive leakage chain rule for quantum min-entropy] \label{lemma:interactiveLCR}
Suppose $\rho_{A_0B_0} $ is a cq-state,
  where $A_0$ is classical. Let $\Pi =  \{\rA,\rB\}$ be an $(r,m_A,m_B)$ two-party protocol with classically-controlled quantum operations $\Phi_i$ controlled by $A_0$.
 Let  $\sigma_{A_0A_rB_r}= \left[\rA\circledast\rB\right](\rho_{A_0B_0})$ be the final state of the protocol. 
 Then
	\begin{align}
	&H_{\min}(A_0|B_r)_{\sigma}\geq H_{\min}(A_0|B_0)_\rho -  \min\{ m_A+m_B, 2m_{A}\}, \label{eq:LCR}
	\end{align}
 We say that $\sigma_{B_r}$ is an \emph{interactive leakage} of $A_0$ generated by $\Pi$.

\end{theorem}

%
\begin{proof}


Suppose $\lambda=\log d_{A} - \Hmin(A|B_0)_\rho$.
  By definition~(\ref{def:Hmin}) there exists a density operator $\tau_{{B_0}}$
   such that
  $$\rho_{A_0B_0} \leq 2^{\lambda} \frac{\I_{A_0}}{d_{A_0}}\ot \tau_{B_0}.$$
Suppose $\ket{\xi}_{B_0E}$ is a purification of $\tau_{B_0}$ over $\cB_0\ot\cE$.
  Without loss of generality, we assume that Alice and Bob have auxiliary quantum systems $R_1,R_2$, respectively,  initialized in $\ket{0}_{R_1},\ket{0}_{R_2}$, so that
  the protocol $\Pi$ can be extended to a protocol $\tilde{\Pi}$
  such that the quantum operations of $\tilde{\Pi}$ are unitary operators controlled by $A_0$ for $\rA$
  and unitaries for $\rB$,
  and $\tr_{R_1R_2}\left(\tilde{\Pi}(\rho_{A_0B_0}\otimes \ket{0}_{R_1R_2}\bra{0})\right)= \Pi(\rho_{A_0B_0})$.
Now initially we have
    $$\rho_{{A_0}{B_0}R_1R_2} \leq  \frac{2^{\lambda}}{d_{{A_0}}}  \sum_a \ket{a}_{A_0}\bra{a} \ot \tr_E\left(\ket{\xi}_{{B_0}E}\bra{\xi}\right)\otimes\ket{0}_{R_1R_2}\bra{0}.$$
After the protocol the inequality becomes
    \begin{align*}
    \sigma_{{A_0}A_r{B_r}R_1R_2} \leq&  \frac{2^{\lambda}}{d_{{A_0}}}  \sum_a  \tilde{\Pi} \left(\ket{a}_{A_0}\bra{a} \ot \tr_E\left(\ket{\xi,0}_{{B_0}ER_2}\bra{\xi,0}\right)\otimes \ket{0}_{R_1}\bra{0} \right)\\
    =& \frac{2^{\lambda}}{d_{{A_0}}}  \sum_a   \tr_E \left( \tilde{\Pi}\ot \I_E \left(\ket{a}_{A_0}\bra{a} \ot  \ket{\xi,0}_{{B_0}ER_2}\bra{\xi,0}\ot \ket{0}_{R_1}\bra{0}\right)\right)\\
   \stackrel{(a)}{=}& \frac{2^{\lambda}}{d_{{A_0}}}  \sum_a   \ket{a}_{A_0}\bra{a} \ot  \tr_E \left( \sum_{i=1}^{2^{m_A+m_B}} \lambda_i^a \ket{\xi_i}_{{B_r}ER_2}\ot \ket{\zeta_i}_{A_rR_1}
   \sum_{j=1}^{2^{m_A+m_B}} \lambda_j^a \bra{\xi_j}_{{B_r}ER_2}\ot \bra{\zeta_j}_{A_rR_1} \right),
    \end{align*}
    where $(a)$ follows from Lemma~\ref{lemma:Yao} and the coefficients $\lambda_j^a$ depend on the classical   $a$. Consequently,
    \begin{align*}
    \sigma_{{A_0}{B}_r}=\tr_{A_rR_1R_2} \sigma_{{A_0}A_r{B}_r{R_1R_2}} \leq& \frac{2^{\lambda}}{d_{{A_0}}}  \sum_a   \ket{a}_{A_0}\bra{a} \ot  \tr_{ER_2} \left( \sum_{i=1}^{2^{m_A+m_B}} \left(\lambda_i^a\right)^2  \ket{\xi_i}_{{B_r}ER_2}\bra{\xi_i}\right)\\
    \leq& \frac{2^{\lambda +m_A+m_B}}{d_{{A_0}}}  \sum_a   \ket{a}_{A_0}\bra{a} \ot  \tr_{ER_2} \left(\frac{1}{2^{m_A+m_B}} \sum_{i=1}^{2^{m_A+m_B}}   \ket{\xi_i}_{{B_r}ER_2}\bra{\xi_i}\right)\\
    =&\frac{2^{\lambda +m_A+m_B}}{d_{{A_0}}}  \sum_a   \ket{a}_{A_0}\bra{a} \ot \omega_{B_r},
    \end{align*}
    where $\omega_{B_r}=  \tr_{ER_2} \left(\frac{1}{2^{m_A+m_B}} \sum_{i=1}^{2^{m_A+m_B}}   \ket{\xi_i}_{{B_r}ER_2}\bra{\xi_i}\right)$.
    Therefore, we have, by Definition~\ref{def:minentropy},
    \[
    \Hmin({A_0}|{B}_r)_{\sigma}\geq \Hmin({A_0}|{B_0})_\rho -( m_{B}+m_{A}).
    \]

Each round of the interactive protocol consists of the following steps:
\begin{enumerate}
  \item Bob performs a unitary operation on his qubits.
  \item Bob sends some qubits   to Alice.
  \item Alice performs a (classical-controlled) quantum operation on her qubits.
  \item Alice sends some qubits to Bob.
\end{enumerate}
Note that only when Alice sends qubits Bob does the min-entropy change
and by  Lemma~\ref{lemma:LCR2}, the entropy decreases by at most two for each qubit that Alice sends to Bob. Thus, we have
    \begin{align*}
    &\Hmin(A_0|B_r)_{\sigma}\geq \Hmin(A_0|B_0)_\rho -  2m_{A}. %
    \end{align*}

\end{proof}

In fact, interactive leakage chain rule can be strengthened to allow pre-shared entanglement between
Alice and Bob by considering only the one-way communication complexity from Alice to Bob.
\bt[Interactive leakage chain rule for quantum min-entropy with pre-shared entanglement] \label{lemma:interactiveLCR2}
  Suppose Alice and Bob share an initial state  $\rho_{A_0B_0}=\ket{\Phi^+}_{A_0''B_0''}^{\otimes m}\bra{\Phi^+}^{\otimes m}\otimes \rho_{A_0'B_0'}$,
  where $\cA_0=\cA_0'\otimes\cA_0''$, $\cB_0=\cB_0'\otimes\cB_0''$, $\ket{\Phi^+}^{\otimes m}$ are EPR pairs, and $\rho_{A_0'B_0'}$ is a cq state.
  If  an $(r, m_A, m_B)$ two-party    interactive protocol $\Pi$, where the quantum operations  for $\rA$ are classically controlled by $A_0$,
 is  executed by Alice and Bob with $m_A\leq m$,
then
    \begin{align}
    &\Hmin(A_0|B_r)_{\sigma}\geq \Hmin(A_0|B_0)_\rho -  2m_{A}, %
    \end{align}
where $\sigma_{A_0B_r}= \tr_{A_r}\left[\rA\circledast\rB\right](\rho_{A_0B_0})$.
\et


\subsection{Communication Lower Bound} \label{sec:Comm}
In the problem of classical communication over (two-way) quantum channels, Alice wishes to send $n$  classical bits $X$ to Bob,
who then applies a quantum measurement and observes outcome $Y$. The famous Holevo theorem~\cite{Hol73} established a lower bound that
the mutual information between $X$ and $Y$ is at most $m$ if $m$ qubits are sent from Alice to Bob.
Cleve \etal extended the Holevo theorem to interactive protocols~\cite[Theorem 2]{CDNT99}: for Bob to acquire $m$ bits of mutual information, Alice has to send at least $m/2$ qubits to Bob
and the two-way communication complexity is at least $m$ qubits.
Nayak and Salzman further improved these results in that Bob only recovers $X$ with probability $p$~\cite{NS06}.

Herein we provide another version of the classical communication lower bound.
Our results are more general since we allow the initial shared states to be separable.

\bc \label{cor:comm_bound}
  Suppose Alice and Bob share a  cq  state $\rho = \rho_{A_0B_0}=\sum_{a}p_a \ket{a}_{A_0}\bra{a}\otimes \rho_{B_0}^{a} \in D(\cA_0\ot\cB_0)$,
  where Alice holds system $A_0$ of classical information and Bob holds system $B_0$. 
Suppose Alice wants to send $a$ to Bob by  an $(r, m_A, m_B)$ interactive protocol $\Pi$
such that Bob can recover $a$ with probability at least $p\in(0,1]$.
Then
    \begin{align}
    m_{B}+m_{A} &\geq \Hmin(A_0|B_0)_\rho- \log \frac{1}{p}; \label{eq:abn}\\
    2m_{A} &\geq  \Hmin(A_0|B_0)_\rho- \log \frac{1}{p}.\label{eq:2an}
    \end{align}

\ec

\begin{remark}
A protocol that uses the superdense coding techniques~\cite{BW92} can achieve Eqs.~(\ref{eq:abn}) and (\ref{eq:2an}) with equalities.
\end{remark}

\begin{remark}
As an application, we can recover the communication lower bounds by Nayak and Salzman~\cite[Theorems 1.1 and 1.3]{NS06}\footnote{Nayak and Salzman have another stronger result \cite[Theorems 1.4]{NS06} when there is no initial correlation between Alice and Bob.} when $\Hmin(A_0|B_0)_\rho=n$, where $A_0$ is of $n$ bits.
Note that they did a round reduction argument by using Yao's lemma so that the two-party protocol can be simulated by Alice sending a \emph{single} message of length $(m_A+m_B)$ to Bob.
However, this method requires a compression and decompression procedure, which unlikely generalizes to the case with initial correlations.
\end{remark}


\section{Conclusion}
We proved an interactive leakage chain rule for quantum min-entropy and discussed its applications in quantum communication complexity of classical information
and the lower bounds for quantum private information retrieval.
We may also apply our result to other scenarios. For example, our
 we can also derive limitations for information-theoretically secure quantum fully homomorphic encryption~\cite{LC18,NS18,Newman18},
 where the essential ingredient of the proof is  Nayak's bound~\cite{Nayak99}.
 To be more specific, instead of using Nayak's bound, we can use the communication lower (Corollary~\ref{cor:comm_bound}) derived by the interactive leakage chain rule~(Theorem~\ref{lemma:interactiveLCR}) to develop new limitations. This is our ongoing research.

CYL was was financially supported from the Young Scholar Fellowship Program by Ministry of Science and Technology (MOST) in Taiwan, under
Grant MOST107-2636-E-009-005.
KMC was partially supported by 2016 Academia Sinica Career Development Award under Grant
No. 23-17 and the Ministry of Science and Technology, Taiwan under Grant No. MOST 103-2221-
E-001-022-MY3.



\end{document}